\documentclass{sig-alternate-05-2015}

\usepackage[utf8]{inputenc}
\usepackage[english]{babel} 
\usepackage{hyperref} 
\usepackage[noabbrev,capitalize]{cleveref}  
\usepackage{mathtools}  
\usepackage{thmtools} 
\usepackage{thm-restate}  

\newtheorem{theorem}{Theorem}
\newtheorem{lemma}[theorem]{Lemma}
\newtheorem{proposition}[theorem]{Proposition}
\newtheorem{corollary}[theorem]{Corollary}

\newtheorem{definition}{Definition}

\newcounter{term}
\renewcommand*{\theterm}{\Alph{term}}

\AtBeginDocument{%
  \let\mylabel\label
}
\newcommand{\mytag}[1]{%
  \begingroup 
    \refstepcounter{term}%
    \mylabel{#1}%
    \text{\theterm}%
  \endgroup
}

\DeclarePairedDelimiter{\diagfences}{(}{)}
\newcommand{\diag}{\operatorname{diag}\diagfences}

\newcommand{\Reals}{\mathbb{R}}
\newcommand{\Complex}{\mathbb{C}}
\newcommand{\Rationals}{\mathbb{Q}}

\newcommand{\Naturals}{\mathbb{N}}

\newcommand{\myvector}{\boldsymbol}

\newcommand{\NP}{\textbf{NP}}

\newcommand{\PSPACE}{\textbf{PSPACE}}

\begin{document}

\CopyrightYear{2017}
\setcopyright{rightsretained}
\conferenceinfo{HSCC'17}{April 18-20, 2017, Pittsburgh, PA, USA}
\isbn{978-1-4503-4590-3/17/04}
\doi{http://dx.doi.org/10.1145/3049797.3049798}

\title{On the Polytope Escape Problem for\\ Continuous Linear Dynamical Systems}
%
%
%
%
%

\numberofauthors{3} 
%
%
%

\author{
\alignauthor{
Jo\"{e}l Ouaknine\\
\affaddr{MPI-SWS and Oxford U.}\\
\email{joel@mpi-sws.org}}
\alignauthor{
Jo\~{a}o Sousa-Pinto\\
\affaddr{Oxford U.}\\
\email{jspinto@cs.ox.ac.uk}}
\alignauthor{
James Worrell\\
\affaddr{Oxford U.}\\
\email{jbw@cs.ox.ac.uk}}
}

\maketitle

\begin{abstract}

  The Polytope Escape Problem for continuous linear dynamical
  systems consists of deciding, given an affine function
  $f:\mathbb{R}^{d}\rightarrow \mathbb{R}^{d}$ and a convex polytope
  $\mathcal{P}\subseteq\mathbb{R}^d$, both with rational descriptions,
  whether there exists an initial point
  $\boldsymbol{x}_0$ in $\mathcal{P}$ such that the trajectory of the unique
  solution to the differential equation
\begin{equation*}
\begin{displaystyle} \begin{cases}
\dot{\boldsymbol{x}}(t)=f(\boldsymbol{x}(t)) \\
\boldsymbol{x}(0)=\boldsymbol{x}_{0}
\end{cases} \end{displaystyle}
\end{equation*}
is entirely contained in $\mathcal{P}$.  We show that this problem is
reducible in polynomial time to the decision version of linear
programming with real algebraic coefficients.  The latter is a special
case of the decision problem for the existential theory of real closed
fields, which is known to lie between \NP{} and
\PSPACE{}.  Our algorithm makes use of spectral techniques and
relies, among others, on tools from Diophantine approximation.
\end{abstract}




\section{Introduction}

In ambient space $\Reals^{d}$, a \emph{continuous linear
  dynamical system} is a trajectory $\myvector{x}(t)$, where $t$
ranges over the non-negative reals, defined by a differential equation
$\dot{\myvector{x}}(t)=f(\myvector{x}(t))$ in which the function
$f$ is \emph{affine} or \emph{linear}. If the initial point
$\myvector{x}(0)$ is given, the differential equation uniquely
defines the entire trajectory. (Linear) dynamical systems have been
extensively studied in Mathematics, Physics, and Engineering, and more
recently have played an increasingly important role in Computer
Science, notably in the modelling and analysis of cyber-physical
systems; a recent and authoritative textbook on the matter
is~\cite{Alu15}.

In the study of dynamical systems, particularly from the perspective
of control theory, considerable attention has been given to the study
of \emph{invariant sets}, i.e., subsets of $\Reals^{d}$ from which
no trajectory can escape; see, e.g.,
\cite{CastelanH92,BlondelT00,BM07,SDI08}. Our focus in the present
chapter is on sets with the dual property that \emph{no trajectory
  remains trapped}. Such sets play a key role in analysing
\emph{liveness} properties in cyber-physical systems (see, for
instance,~\cite{Alu15}): discrete progress is ensured by
guaranteeing that all trajectories (i.e., from any initial starting
point) must eventually reach a point at which they `escape'
(temporarily or permanently) the set in question.

More precisely, given an affine function
$f:\Reals^{d}\rightarrow \Reals^{d}$ and a convex polytope
$\mathcal{P}\subseteq\Reals^{d}$, both specified using rational
coefficients encoded in binary, we consider the \emph{Polytope
  Escape Problem} which asks whether there is some point
$\myvector{x}_0$ in $\mathcal{P}$ for which the corresponding
trajectory of the solution to the differential equation
\begin{equation*}
\begin{displaystyle} \begin{cases}
\dot{\myvector{x}}(t)=f(\myvector{x}(t)) \\
\myvector{x}(0)=\myvector{x}_{0}
\end{cases} \end{displaystyle}
\end{equation*}
is entirely contained in $\mathcal{P}$. Our main result is to show
that this problem is decidable by reducing it in polynomial time to
the decision version of linear programming with real algebraic
coefficients, which itself reduces in polynomial time to deciding the
truth of a sentence in the first-order theory of the reals: a problem
whose complexity is known to lie between $\NP$ and
$\PSPACE$ \cite{Canny88}. Our algorithm makes use of spectral
techniques and relies among others on tools from Diophantine
approximation.

It is interesting to note that a seemingly closely related problem,
that of determining whether a given trajectory of a linear dynamical
system ever hits a given hyperplane (also known as the
\emph{continuous Skolem Problem}), is not known to be decidable; see,
in particular,~\cite{ContinuousSkolem,ContinuousSkolem3,COW16b:LICS16}. When the
target is instead taken to be a single point (rather than a
hyperplane), the corresponding reachability question (known as the
\emph{continuous Orbit Problem}) can be decided in polynomial
time~\cite{Hainry08}.

\section{Mathematical Background}

\subsection{Kronecker's Theorem}
Let $\mathbb{T}$ denote the group of complex numbers of modulus $1$,
with multiplication as group operation.  Then the function
$\phi:\Reals \rightarrow \mathbb{T}$ given by
$\phi(x)=\exp(2\pi i x)$ is a homomorphism from the additive
group of real numbers to $\mathbb{T}$, with kernel the subgroup of
integers.

Recall from~\cite{HardyAndWright} the following classical theorem of
Kronecker on simultaneous inhomogeneous Diophantine approximation.
\begin{theorem}[Kronecker]
  Let $\theta_1,\ldots,\theta_s$ be real numbers such that the set
  $\{\theta_1,\ldots,\theta_s,1\}$ is linearly independent over
  $\mathbb{Q}$.  Then for all $\psi_1,\ldots,\psi_s \in \Reals$
  and $\varepsilon > 0$, there exists a positive integer $n$
and integers $n_1,\ldots,n_s$ such that
\[ |n\theta_1 - \psi_1 - n_1| < \varepsilon, \ldots ,
   |n\theta_s - \psi_s - n_s| < \varepsilon \, .\]
\end{theorem}

We obtain the following simple corollary:
\begin{corollary}
  Let $\theta_1,\ldots,\theta_s$ be real numbers such that the set
  $\{ \theta_1,\ldots,\theta_s,1\}$ is linearly independent over
  $\mathbb{Q}$.  Then
\[ \{ (\phi(n\theta_1),\ldots,\phi(n\theta_s)) : n \in \Naturals \} \]
is a dense subset of $\mathbb{T}^s$.
\label{corl:kronecker}
\end{corollary}
\begin{proof}
  Since $\phi$ is surjective, an arbitrary element of $\mathbb{T}^s$
  can be written in the form $(\phi(\psi_1),\ldots,\phi(\psi_s))$ for
  some real numbers $\psi_1,\ldots,\psi_s$.
Applying Kronecker's Theorem, we get that
for all $\varepsilon > 0$, there exists a positive integer $n$
and integers $n_1,\ldots,n_s$ such that
\[ |n\theta_1 - \psi_1 - n_1| < \varepsilon, \ldots , |n\theta_s
  - \psi_s - n_s| < \varepsilon \, .\] By continuity of $\phi$ it
follows that $(\phi(\psi_1),\ldots,\phi(\psi_s))$ is a limit point of
$ \{ (\phi(n\theta_1),\ldots,\phi(n\theta_s)) : n \in \Naturals \}$.
This establishes the result.
\end{proof}

\subsection{Laurent polynomials}

A multivariate \emph{Laurent polynomial} is a polynomial in positive
and negative powers of variables $z_1,\ldots,z_s$ with complex
coefficients.  We are interested in Laurent polynomials of the special form
\[ g = \sum_{j=1}^k \left( c_j {z_1}^{n_{1,j}}\ldots {z_s}^{n_{s,j}} +
    \overline{c_j} {z_1}^{-n_{1,j}}\ldots {z_s}^{-n_{s,j}} \right) \,
  ,\] where $c_1,\ldots,c_k \in \Complex$ and
$n_{1,1},\ldots,n_{s,k} \in \mathbb{Z}$.  We call such $g$
\emph{self-conjugate Laurent polynomials}.  Notice that if
$a_1,\ldots,a_s \in \mathbb{T}$ then $g(a_1,\ldots,a_s)$ is a real
number, so we may regard $g$ as a function from $\mathbb{T}^s$ to
$\Reals$.

\begin{lemma}
\label{lem:first_bound}
Let $g \in \Complex[z^{\pm 1}_1,\ldots,z^{\pm 1}_s]$ be a
self-conjugate Laurent polynomial that has no constant term.
Given real numbers $\theta_1,\ldots,\theta_s \in \Reals$ such that
$\theta_1,\ldots,\theta_s,1$ are linearly independent over
$\mathbb{Q}$, define a function
$f:\Reals_{\geq 0}\rightarrow \Reals$ by
\[ f(t) = g(\phi(t\theta_1),\ldots,\phi(t\theta_s)) \, .\]
Then either $f$ is identically zero, or
\begin{align*}
\liminf\limits_{n\rightarrow\infty} f(n) < 0 \, ,
\end{align*}
where $n$ ranges over the nonnegative integers.
\end{lemma}

\begin{proof}
  Recall that we may regard $g$ as a function from
$\mathbb{T}^s$ to $\Reals$.  Now we consider the
  function $g \circ \phi^s : \Reals^s \rightarrow \Reals$,
  \[ (x_1,\ldots,x_s) \mapsto g(\phi(x_1),\ldots,\phi(x_s)) \, . \] We
  use an averaging argument to establish that either $g\circ\phi^s$ is
  identically zero on $\Reals^s$ or there exist
  $x_1^*,\ldots,x_s^* \in [0,1]$ such that
  $g(\phi(x_1^*),\ldots,\phi(x_s^*))<0$.

  Since $\int_0^{1} \exp(2\pi i n x) dx=0$ for all non-zero
  integers $n$, it holds that
\begin{equation*}
\int_0^{1} \ldots \int_0^{1}
g(\phi(x_1),\ldots,\phi(x_s))
dx_1 \ldots dx_s =0 \, .
\end{equation*}
Suppose that
$g\circ \phi^s$ is not identically zero over
$\Reals^s$ and hence not identically zero over $[0,1]^{s}$.
Then $g\circ \phi^s$ cannot be nonnegative on $[0,1]^{s}$, since
the integral over a set of positive measure of a continuous
nonnegative function that is not identically zero must be strictly
positive.  We conclude that there must exist
$(x_1^*,\ldots,x_s^*) \in [0,1]^{s}$ such that
$g(\phi(x_1^*),\ldots,\phi(x_s^*))<0$.

By assumption, $\theta_1,\ldots,\theta_s,1$ are linearly
independent over $\mathbb{Q}$.  By~\cref{corl:kronecker} it
follows that
\[ \{ (\phi(n\theta_1),\ldots,\phi(n\theta_s)) : n \in
  \Naturals \} \] is dense in $\mathbb{T}^s$ and hence has
$(\phi(x_1^*),\ldots,\phi(x_s^*))$ as a limit point.
Since $g\circ\phi^s$ is continuous, there are
arbitrarily large $n\in\Naturals$ for which
\[
f(n) = g(\phi(n\theta_1),\ldots,\phi(n\theta_s))
    \leq \textstyle\frac{1}{2} g(\phi(x_1^*),\ldots,\phi(x_s^*)) < 0 \, , \]
which proves the result.
\end{proof}

Note that this proof could be made constructive by using an effective
version of Kronecker's Theorem, as studied in
\cite{ConstructiveKronecker1} and \cite{ConstructiveKronecker2},
although we do not make use of this fact in the present paper.

We say that a self-conjugate Laurent polynomial $g$ is \emph{simple}
if it has no constant term and each monomial mentions only a single
variable.  More precisely, $g$ is simple if it can be written in the
form 
\[ g=\sum_{j=1}^k c_j z_{i_j}^{n_j} + \overline{c_j} z_{i_j}^{-n_j} \,
  ,\] where $c_1,\ldots,c_k \in \mathbb{C}$,
$i_1,\ldots,i_k \in \{1,\ldots,s\}$, and
$n_1,\ldots,n_k \in \mathbb{Z}$.

The following consequence of~\cref{lem:first_bound} will be key to
proving decidability of the problem at hand. It is an extension of Lemma 4 from~\cite{Bra06}.

\begin{theorem}
\label{thm:liminf}
Let $g\in\Complex[z_1^{\pm 1},\ldots,z_s^{\pm 1}]$ be a simple
self-conjugate Laurent polynomial and $\theta_1,\ldots,\theta_s$
non-zero real numbers.  Then either
\begin{align*}
& g(\phi(t\theta_1),\ldots,\phi(t\theta_s)) = 0 \mbox{ for all $t\in\Reals$}
\intertext{or}
&
\liminf\limits_{n\rightarrow\infty}  g(\phi(n\theta_1),\ldots,\phi(n\theta_s)) < 0 \, ,
\end{align*}
where $n$ ranges over the nonnegative integers.
\end{theorem}

\begin{proof}
  Note that if $1, \theta_{1}, \ldots, \theta_{s}$ are linearly
  independent over $\Rationals$ then the result follows from
  \cref{lem:first_bound}.  Otherwise, let
  $\lbrace \theta_{i_{1}}, \ldots, \theta_{i_{k}} \rbrace$ be a
  maximal subset of $\lbrace \theta_{1}, \ldots, \theta_{s} \rbrace$
  such that $1, \theta_{i_{1}}, \ldots, \theta_{i_{k}}$ are linearly
  independent over $\Rationals$.

Then, for some $N\in\Naturals$ and each $j$, one can write
\begin{equation*}
N\theta_{j}= \left( m  +\sum\limits_{l=1}^{k} n_{l}\theta_{i_{l}} \right) \, ,
\end{equation*}
where $m,n_{1},\ldots,n_{k}$ are integers that depend on $j$, whilst $N$ does not depend on $j$.  It follows that for all $j$ and $t \in \Reals$,
\begin{align*}
\phi(N\theta_{j} t) &= \phi(m t) \cdot \prod\limits_{l=1}^{k} \phi(n_{l} \theta_{i_{l}} t) \\
&= \phi(t)^{m} \cdot \prod\limits_{l=1}^{k} \phi(\theta_{i_{l}} t)^{n_{l}}  \, .
\end{align*}
In other words, for all $j \geq k+1$, $\phi(N \theta_{j} t)$ can be written as a product of positive and negative powers of the terms
\begin{equation*}
  \phi(t), \phi(\theta_{i_{1}} t), \ldots, \phi(\theta_{i_{k}} t) \, .
\end{equation*}
It follows that there exists a self-conjugate Laurent polynomial
$h\in\Complex [z_1^{\pm 1},\ldots,z_k^{\pm 1}]$, not necessarily
simple, but with zero constant term, such that for all
$t\in \Reals$,
\[ g(\phi(N \theta_{1} t), \ldots, \phi(N \theta_{s} t)) =
  h(\phi(\theta_{i_{1}} t),\ldots,\phi(\theta_{i_{k}} t)) \, .\]
Since $1, \theta_{i_{1}}, \ldots, \theta_{i_{k}}$ are linearly independent over $\Rationals$, the result follows by applying \cref{lem:first_bound}
to $h$.
\end{proof}

\subsection{Jordan Canonical Forms}

Let $A \in \mathbb{Q}^{d \times d}$ be a square matrix with rational
entries.  The \emph{minimal polynomial} of $A$ is the unique monic
polynomial $m(x) \in \mathbb{Q}[x]$ of least degree such that
$m(A)=0$.  By the Cayley-Hamilton Theorem the degree of $m$ is at most
the dimension $d$ of $A$. The set $\sigma(A)$ of eigenvalues is the set of
roots of $m$.  The \emph{index} of an eigenvalue $\lambda$, denoted
by $\nu(\lambda)$, is its multiplicity as a root of $m$. We
use $\nu(A)$ to denote $\max_{\lambda\in\sigma(A)} \nu(\lambda)$: the
maximum index over all eigenvalues of $A$.  Given an eigenvalue
$\lambda \in \sigma(A)$, we say that $\myvector{v} \in \Complex^d$
is a \emph{generalised eigenvector} of $A$ if
$\myvector{v}\in \ker(A-\lambda I)^{k}$, for some $k\in\Naturals$.

We denote by $\mathcal{V}_{\lambda}$ the subspace of $\Complex^d$
spanned by the set of generalised eigenvectors associated with some
eigenvalue $\lambda$ of $A$. We denote the subspace of $\Complex^d$
spanned by the set of generalised eigenvectors associated with some
real eigenvalue by $\mathcal{V}^{r}$.  We likewise denote the subspace
of $\Complex^d$ spanned by the set of generalised eigenvectors
associated to eigenvalues with non-zero imaginary part by
$\mathcal{V}^{c}$.

It is well known that each vector $\myvector{v}\in\Complex^{d}$
can be written uniquely as
$\myvector{v}=\displaystyle{
  \sum\limits_{\lambda\in\sigma(A)}\myvector{v}_{\lambda}}$,
where $\myvector{v}_{\lambda}\in\mathcal{V}_{\lambda}$.
It follows that $\myvector{v}$ can also be uniquely written as
$\myvector{v}=\myvector{v}^{r}+\myvector{v}^{c}$, where
$\myvector{v}^{r} \in\mathcal{V}^{r}$ and
$\myvector{v}^{c} \in\mathcal{V}^{c}$.

We can write any matrix $A\in \mathbb{C}^{d\times d}$ as $A=Q^{-1}JQ$
for some invertible matrix $Q$ and block diagonal Jordan matrix
$J=\diag{J_{1},\ldots,J_{N}}$, with each block $J_{i}$ having the
following form:

\begin{equation*}
\begin{pmatrix}
\lambda	&&	1		&&	0		&&	\cdots	&&	0		\\
0		&&	\lambda	&&	1		&&	\cdots	&&	0		\\
\vdots	&&	\vdots	&&	\vdots	&&	\ddots	&&	\vdots	\\
0		&&	0		&&	0		&&	\cdots	&&	1		\\
0		&&	0		&&	0		&&	\cdots	&&	\lambda	\\
\end{pmatrix}
\end{equation*}
Given a rational matrix $A$, its Jordan Normal Form $A=Q^{-1}JQ$ can be
computed in polynomial time, as shown in \cite{Cai94}.

Note that each vector $\myvector{v}$ appearing as a column of the
matrix $Q^{-1}$ is a generalised eigenvector of $A$. We also note that the
index $\nu(\lambda)$ of some eigenvalue $\lambda$ corresponds to the
dimension of the largest Jordan block associated with it.

One can obtain a closed-form expression for powers of block diagonal
Jordan matrices, and use this to get a closed-form expression for
exponential block diagonal Jordan matrices. In fact, if $J_{i}$ is a
$k\times k$ Jordan block associated with some eigenvalue $\lambda$,
then
\noindent
\begin{equation*}
J_{i}^{n}=\begin{pmatrix}
\lambda^{n}	&&	n\lambda^{n-1}	&&	{n\choose 2}\lambda^{n-1}	&&
\cdots		&&	{n\choose k-1}\lambda^{n-k+1}				\\
0			&&	\lambda^{n}		&&	n\lambda^{n-1}				&&
\cdots		&&	{n\choose k-2}\lambda^{n-k+2}				\\
\vdots	&&	\vdots	&&	\vdots	&&	\ddots	&&	\vdots			\\
0		&&	0		&&	0		&&	\cdots	&&	n\lambda^{n-1}	\\
0		&&	0		&&	0		&&	\cdots	&&	\lambda^{n}		\\
\end{pmatrix}
\end{equation*}
and
\begin{equation*}
\exp(J_{i}t)=\exp(\lambda t) \begin{pmatrix}
1		&&	t		&&	\cdots	&&	\frac{t^{k-1}}{(k-1)!}	\\
0		&&	1		&&	\cdots	&&	\frac{t^{k-2}}{(k-2)!}	\\
\vdots	&&	\vdots	&&	\ddots	&&	\vdots						\\
0		&&	0		&&	\cdots	&&	t							\\
0		&&	0		&&	\cdots	&&	1							\\
\end{pmatrix}
\end{equation*}


In the above, ${n\choose j}$ is defined to be $0$ when $n<j$.

\begin{proposition}
  Let $\myvector{v}$ lie in the generalised eigenspace
  $\mathcal{V}_{\lambda}$ for some $\lambda \in \sigma(A)$.  Then
  $\myvector{b}^{T}\exp(At)\myvector{v}$ is a linear combination
  of terms of the form $t^{n}\exp(\lambda t)$.
\label{prop:linear}
\end{proposition}

\begin{proof}
  Note that, if $A=Q^{-1}JQ$ and $J=diag(J_{1},\ldots,J_{N})$ is a
  block diagonal Jordan matrix, then
  \[\exp(At)=Q^{-1}\exp(Jt)Q\]
  and
  \[\exp(Jt)=\diag{\exp(J_{1}t),\ldots,\exp(J_{N}t)} \, .\]
The result follows by observing that $Q\myvector{v}$ is zero in every component
other than those pertaining the block corresponding to the eigenspace
$\mathcal{V}_{\lambda}$.
\end{proof}

In order to compare the asymptotic growth of expressions of the form
$t^{n}\exp(\lambda t)$, for $\lambda\in\Reals$ and
$n\in\Naturals_0$, we define $\prec$ to be the lexicographic order on
$\Reals\times\Naturals_{0}$, that is,
\begin{equation*}
(\eta,j)\prec (\rho,m) \quad \mbox{iff} \quad
\eta<\rho \mbox{ or } (\eta = \rho \mbox{ and } j< m) \, .
\end{equation*}
Clearly $\exp(\eta t)t^{j}=o(\exp(\rho t)t^{m})$ as $t\rightarrow \infty$ if and only if $(\eta,j)\prec (\rho,m)$.

\begin{definition}
If $\myvector{b}^{T}\exp(At)\myvector{v}$ is not identically zero,
the maximal $(\rho,m)\in\Reals\times\Naturals_{0}$ (with respect
to $\prec$) for which there is a term $t^{m}\exp (\lambda t)$ with
$\Re(\lambda)=\rho$ in the closed-form expression for
$\myvector{b}^{T}\exp(At)\myvector{v}$ is called \emph{dominant} for
$\myvector{b}^{T}\exp(At)\myvector{v}$.
\end{definition}

Before we can proceed, we shall need the following auxiliary result:

\begin{proposition}
\label{conj-relation}
Suppose that $\myvector{v}\in\Reals^{d}$ and that \[\myvector{v}=\sum\limits_{\lambda\in\sigma(A)} \myvector{v}_{\lambda} \, , \]
where $\myvector{v}_{\lambda} \in\mathcal{V}_{\lambda}$. Then $\myvector{v}_{\overline{\lambda}}$ and $\myvector{v}_{\lambda}$ are component-wise complex conjugates.
\end{proposition}

\begin{proof}
We start by observing that
\begin{align}
\myvector{0}=\myvector{v}-\overline{\myvector{v}}=\sum\limits_{\lambda\in \sigma(A)}(\myvector{v}_{\lambda}-\overline{ \myvector{v}_{ \overline{\lambda}}}) \, .
\label{eq:unique}
\end{align}
But if $\myvector{v}_{\lambda}\in \ker(A-\lambda I)^{k}$ then
$\overline{\myvector{v}_{\lambda}} \in \ker(A-\overline{\lambda}
I)^{k}$, and hence
$\overline{ \myvector{v}_{ \overline{\lambda}}} \in
\mathcal{V}_{\lambda}$.  Thus each summand
$\myvector{v}_{\lambda}-\overline{ \myvector{v}_{
    \overline{\lambda}}}$ in (\ref{eq:unique}) lies in
$\mathcal{V}_{\lambda}$.  Since $\mathbb{C}^d$ is a direct sum of the
generalised eigenspaces of $A$, we must have
$\myvector{v}_{\lambda}=\overline{ \myvector{v}_{
    \overline{\lambda}}}$ for all $\lambda \in \sigma(A)$.
\end{proof}

We now derive a corollary of~\cref{thm:liminf}.
\begin{corollary}
\label{cor:liminf}
Consider a function of the form
$h(t)=\myvector{b}^{T}\exp(At) \myvector{v}^{c}$, where
$\myvector{v}^{c}\in\mathcal{V}^{c}$, with
$(\rho,m)\in\Reals\times \Naturals_{0}$ dominant.
If $h(t)\not\equiv 0$, then we have
\begin{align*}
-\infty<\liminf\limits_{t\rightarrow\infty} \frac{h(t)}{\exp(\rho
  t)t^{m}}<0 \, .
\end{align*}
\end{corollary}

\begin{proof}
  Let
  \[\Re(\sigma(A))=\lbrace \eta\in\Reals:
  \eta+i\theta\in\sigma(A),\mbox{ for some }\theta\in\Reals
  \rbrace \,. \]
  For each $\eta\in\Re(\sigma(A))$ define
  $\myvector{\theta}_{\eta}=\lbrace \theta\in\Reals_{>0}:
  \eta+i\theta \in\sigma(A) \rbrace$. By abuse of notation, we also
  use $\myvector{\theta}_{\eta}$ to refer to the vector whose
  coordinates are exactly the members of this set, ordered in an
  increasing way. We note that, due to~\cref{conj-relation}
  and~\cref{prop:linear}, the
  following holds:

\begin{align*}
\myvector{b}^{T}\exp(At)\myvector{v}^{c} &= \myvector{b}^{T} \exp(At) \sum\limits_{\eta\in\Re(\sigma(A))} \sum\limits_{\theta\in\myvector{\theta}_{\eta}} \myvector{v}_{\eta+i\theta}+\myvector{v}_{\eta-i\theta} \\
&= \sum\limits_{\eta\in\Re(\sigma(A))} \sum\limits_{\theta\in\myvector{\theta}_{\eta}} \myvector{b}^{T} \exp(At) \myvector{v}_{\eta+i\theta} \\
& \qquad \qquad \qquad \qquad \qquad + \overline{\myvector{b}^{T} \exp(At) \myvector{v}_{\eta+i\theta}} \\
& = \sum\limits_{\eta\in\Re(\sigma(A))} \sum\limits_{j=0}^{\nu(A)-1} t^{j}\exp(\eta t)  g_{(\eta,j)}( \exp(i\myvector{\theta}_{\eta}t) )
\end{align*}
for some simple self-conjugate Laurent polynomials
$g_{(\eta,j)}$.
Note that
\begin{equation*}
(\rho,m)=\max\limits_{\prec} \lbrace (\eta,j)\in\Reals\times
\Naturals_{0}: g_{(\eta,j)}(\exp(i\myvector{\theta}_{\eta} t)) \not
\equiv 0 \rbrace \, .
\end{equation*}

The result then follows from~\cref{thm:liminf} and the fact
that
\begin{align*}
\liminf\limits_{t\rightarrow\infty} \frac{h(t)}{\exp(\rho t)t^{m}}=
  \liminf\limits_{t\rightarrow\infty}
  g_{(\rho,m)}(\exp(i\myvector{\theta}_{\rho} t)) \, .
\end{align*}
\end{proof}

\subsection{Computation with Algebraic Numbers}

In this section, we briefly explain how one can represent and manipulate algebraic numbers efficiently.

Any given algebraic number $\alpha$ can be represented as a tuple
$(p,a,\varepsilon)$, where $p\in\mathbb{Q}[x]$ is its minimal
polynomial, $a=a_1+a_2i$, with $a_1,a_2\in\mathbb{Q}$, is an
approximation of $\alpha$, and $\varepsilon \in \mathbb{Q}$ is
sufficiently small that $\alpha$ is the unique root of $p$ within
distance $\varepsilon$ of $a$.  This is referred to as the standard or
canonical representation of an algebraic number.

Let $f\in\mathbb{Z}[x]$ be a polynomial.  The following
root-separation bound, due to Mignotte \cite{Mig82}, can be used to
give a value of $\varepsilon$ such that any disk of radius
$\varepsilon$ in the complex plane contains at most one root of $f$.
\begin{proposition}
Let $f\in\mathbb{Z}[x]$. If $\alpha_{1}$ and $\alpha_{2}$ are distinct roots of $f$, then
\begin{align*}
\lvert \alpha_{1}-\alpha_{2} \rvert > \frac{\sqrt{6}}{d^{(d+1)/2}H^{d-1}}
\end{align*}
where $d$ and $H$ are respectively the degree and height (maximum
absolute value of the coefficients) of $f$.
\end{proposition}
It follows that in the canonical representation $(p,a,\varepsilon)$ of
an algebraic number $\alpha$, where $p$ has degree $d$ and height $H$,
we may choose $a_1,a_2,\epsilon$ to have bit length polynomial in $d$
and $\log H$.

Given canonical representations of two algebraic numbers $\alpha$ and
$\beta$, one can compute canonical representations of $\alpha+\beta$,
$\alpha\beta$, and $\alpha/\beta$, all in polynomial time.  More
specifically, one can:
\begin{itemize}
\item factor an arbitrary polynomial with rational coefficients as a
  product of irreducible polynomials in polynomial time using the LLL
  algorithm, described in \cite{LenstraLenstraLovasz1982};
\item compute an approximation of an arbitrary algebraic number
  accurate up to polynomially many bits in polynomial time, due to the
  work in~\cite{Pan97};
\item use the sub-resultant algorithm (see Algorithm 3.3.7 in
  \cite{Coh93}) and the two aforementioned procedures to compute
  canonical representations of sums, differences, multiplications, and
  quotient of two canonically represented algebraic numbers.
\end{itemize}

\section{Existential First-Order Theory of the Reals}

Let $\myvector{x}=(x_1,\ldots,x_m)$ be a list of $m$ real-valued
variables, and let $\sigma(\myvector{x})$ be a Boolean combination of
atomic predicates of the form $g(\myvector{x})\sim 0$, where each
$g(\myvector{x})$ is a polynomial with integer coefficients in the
variables $\myvector{x}$, and $\sim$ is either $>$ or $=$. Tarski has
famously shown that we can decide the truth over the field $\Reals$ of
sentences of the form
$\phi=Q_1 x_1 \cdots Q_m x_m \sigma(\myvector{x})$, where $Q_i$ is
either $\exists$ or $\forall$. He did so by showing that this theory
admits quantifier elimination (Tarski-Seidenberg Theorem
\cite{Tar51}). The set of all true sentences of such form is called
the first-order theory of the reals, and the set of all true sentences
where only existential quantification is allowed is called the
existential first-order theory of the reals. The complexity class
$\exists\Reals$ is defined as the set of problems having a
polynomial-time many-one reduction to the existential theory of the
reals. It was shown in \cite{Canny88} that
$\exists\Reals\subseteq \mathit{PSPACE}$.

We also remark that our standard representation of algebraic numbers
allows us to write them explicitly in the first-order theory of the
reals, that is, given $\alpha\in\mathbb{A}$, there exists a sentence
$\sigma(x)$ such that $\sigma(x)$ is true if and only if
$x=\alpha$. Thus, we allow their use when writing sentences in the
first-order theory of the reals, for simplicity.

The decision version of linear programming with canonically-defined algebraic coefficients is in $\exists\Reals$, as the emptiness of a convex polytope can easily be described by a sentence of the form $\exists x_1 \cdots \exists x_n \sigma(\myvector{x})$.

Finally, we note that even though the decision version of linear
programming with rational coefficients is in $\mathit{P}$, allowing
algebraic coefficients makes things more complicated. While it has
been shown in \cite{AdlerB94} that this is solvable in time polynomial
in the size of the problem instance and on the degree of the smallest
number field containing all algebraic numbers in each instance, it
turns out that in the problem at hand the degree of that extension can
be exponential in the size of the input. In other words, the splitting
field of the characteristic polynomial of a matrix can have a degree
which is exponential in the degree of the characteristic polynomial.

\section{The Polytope Escape Problem}

The Polytope Escape Problem for continuous linear dynamical systems
consists of deciding, given an affine function
$f:\Reals^{d}\rightarrow \Reals^{d}$ and a convex polytope
$\mathcal{P}\subseteq\Reals^{d}$, whether there exists an initial point
$\myvector{x}_{0} \in \mathcal{P}$ for which the trajectory of the unique
solution to the differential equation
$\dot{\myvector{x}}(t)=f(\myvector{x}(t)),\myvector{x}(0)=\myvector{x}_{0},
t\geq 0$,
is entirely contained in $\mathcal{P}$.  A starting point
$\myvector{x}_{0}\in\mathcal{P}$ is said to be \emph{trapped} if
the trajectory of the corresponding solution is contained in $\mathcal{P}$,
and \emph{eventually trapped} if the trajectory of the corresponding
solution contains a trapped point. Therefore, the Polytope Escape
Problem amounts to deciding whether a trapped point exists, which in
turn is equivalent to deciding whether an eventually trapped point exists.

The goal of this section is to prove the following result:

\begin{theorem}
  The Polytope Escape Problem is polynomial-time reducible to the
  decision version of linear programming with algebraic coefficients.
\end{theorem}

A $d$-dimensional instance of the Polytope Escape Problem is a pair
$(f,\mathcal{P})$, where $f:\Reals^{d}\rightarrow \Reals^{d}$
is an affine function and $\mathcal{P}\subseteq\Reals^{d}$ is a
convex polytope. In this formulation we assume that all numbers
involved in the definition of $f$ and $\mathcal{P}$ are
rational.\footnote{The assumption of rationality is required to
  justify some of our complexity claims (e.g., Jordan Canonical Forms
  are only known to be polynomial-time computable for matrices with
  rational coordinates). Nevertheless, our procedure remains valid in
  a more general setting, and in fact, the overall
  $\exists \Reals$ complexity of our algorithm would not be
  affected if one allowed real algebraic numbers when defining problem
  instances.}

An instance $(f,\mathcal{P})$ of the Polytope Escape Problem is said
to be \emph{homogeneous} if $f$ is a linear function and
$\mathcal{P}$ is a convex polytope cone (in particular,
$\myvector{x}\in\mathcal{P},\alpha>0\Rightarrow \alpha\myvector{x}
\in\mathcal{P}$).

The restriction of the Polytope Escape Problem to homogeneous
instances is called the homogeneous Polytope Escape Problem.

\begin{lemma}
  The Polytope Escape Problem is polynomial-time reducible to the
  homogeneous Polytope Escape Problem.
\end{lemma}

\begin{proof}
  Let $(f,\mathcal{P})$ be an instance of the Polytope Escape
  Problem in $\Reals^{d}$, and write
\begin{equation*}
f(\myvector{x})=A\myvector{x}+\myvector{a}\mbox{ and } \mathcal{P}=\lbrace\myvector{x}\in\Reals^{d}: B_{1}\myvector{x}>\myvector{b}_{1} \wedge B_2\myvector{x}\geq \myvector{b}_{2}\rbrace \, .
\end{equation*}
Now define
\begin{align*}
& A'=\begin{pmatrix}A & \myvector{a}\\ \myvector{0}^T & 0\end{pmatrix},
B_{1}'=\begin{pmatrix}B_{1} & -\myvector{b}_{1} \\ \myvector{0}^T & 1\end{pmatrix},
B_{2}'=\begin{pmatrix}B_{2} & -\myvector{b}_{2} \end{pmatrix} \, ,\\
& \mathcal{P}'=
\left\lbrace
\begin{pmatrix}\myvector{x}\\y
\end{pmatrix} \in\Reals^{d+1}:B_{1}'\begin{pmatrix}\myvector{x}\\y
\end{pmatrix}
  > \myvector{0}
  \wedge B_{2}'\begin{pmatrix}\myvector{x}\\y
\end{pmatrix} \geq \myvector{0} \right\rbrace \, ,
\intertext{and}
& g\begin{pmatrix}\myvector{x}\\y
\end{pmatrix} = A'\begin{pmatrix}\myvector{x}\\y
\end{pmatrix} \, .
\end{align*}
Then $(g,\mathcal{P}')$ is a homogeneous instance
of the Polytope Escape Problem.

It is clear that $\myvector{x}(t)$ satisfies the differential equation
$\dot{\myvector{x}}(t)=f(\myvector{x}(t))$ if and only if
$\begin{pmatrix}
\myvector{x}(t)\\1 \end{pmatrix}$ satisfies
the differential equation
$\begin{pmatrix}\dot{\myvector{x}}\\
\dot{y}
\end{pmatrix} = g\begin{pmatrix}\myvector{x}\\y
\end{pmatrix}=  = \begin{pmatrix} A\myvector{x}+y\myvector{a}\\ 0
\end{pmatrix}$.  In general,
in any trajectory $\begin{pmatrix}\myvector{x}\\y
\end{pmatrix}$ that satisfies this last differential equation, the
$y$-component must be constant.

We claim that $(f,\mathcal{P})$ is a positive instance
of the Polytope Escape Problem if and
only if $(g,\mathcal{P}')$ is a positive instance.
Indeed, if the point
$\myvector{x}_0 \in \Reals^{d}$ is trapped
in $(f,\mathcal{P})$
then the point $\begin{pmatrix}\myvector{x}_0\\1
\end{pmatrix}$ is trapped in $(g,\mathcal{P}')$.
Conversely, suppose that $\begin{pmatrix}\myvector{x}_0\\y_0
\end{pmatrix}$ is trapped in  $(g,\mathcal{P}')$.  Then,
since $B_{1}'\begin{pmatrix}\myvector{x}_0\\y_0
\end{pmatrix}
> \myvector{0}$,
we must have $y_0>0$.  Scaling, it follows that
$\begin{pmatrix} y_0^{-1}\myvector{x}_0\\1
\end{pmatrix}$ is also trapped in $(g,\mathcal{P}')$.  This implies
that $y_0^{-1}\myvector{x}_0$ is trapped in $(f,\mathcal{P})$.
%
\end{proof}

We remind the reader that the unique solution of the differential
equation
$\dot{\myvector{x}}(t)=f(\myvector{x}(t)),\myvector{x}(0)=\myvector{x}_{0},t\geq
0$, where $f(\myvector{x})=A\myvector{x}$, is given by
$\myvector{x}(t)=\exp(At)\myvector{x}_{0}$.
In this setting, the sets of trapped and eventually trapped points are, respectively:
\begin{align*}
\mathit{T}&=\lbrace \myvector{x}_0\in\Reals^{d}: \forall t\geq 0, \exp(At)\myvector{x}_{0} \in \mathcal{P}\rbrace \\
\mathit{ET}&=\lbrace\myvector{x}_{0}\in\Reals^{d}:\exists t\geq 0,\exp(At)\myvector{x}_{0}\in\mathit{T}\rbrace
\end{align*}
Note that both $\mathit{T}$ and $\mathit{ET}$ are convex subsets of $\Reals^{d}$.

\begin{lemma}
  The homogeneous Polytope Escape Problem is polynomial-time
  reducible to the decision version of linear programming with
  algebraic coefficients.
\end{lemma}

\begin{proof}
  Let
  $\myvector{x}_{0}=\myvector{x}_{0}^{r}+\myvector{x}_{0}^{c}$,
  where $\myvector{x}_{0}^{r}\in \mathcal{V}^{r}$ and
  $\myvector{x}_{0}^{c}\in \mathcal{V}^{c}$. We start by showing
  that if $\myvector{x}_{0}$ lies in the set $T$ of trapped points
  then its component $\myvector{x}_{0}^{r}$ in the real eigenspace
  $\mathcal{V}^{r}$ lies in the set $\mathit{ET}$ of eventually
  trapped points.
Due to the
fact that the intersection of finitely many convex polytopes is still
a convex polytope, it suffices to prove this claim for the case when
  $\mathcal{P}$ is defined by a single inequality---say
  $\mathcal{P}=\lbrace \myvector{x}\in\Reals^{d}:
  \myvector{b}^{T}\myvector{x}\triangleright 0\rbrace$, where
  $\triangleright$ is either $>$ or $\geq$.

  We may assume that
  $\myvector{b}^{T} \exp(At) \myvector{x}_{0}^{c}$ is not identically
  zero, as in that case
  \[\myvector{b}^{T} \exp(At)\myvector{x}_{0} \equiv
  \myvector{b}^{T} \exp(At)\myvector{x}_{0}^{r}\]
  and our claim holds trivially.
  Also, if $\myvector{x}_{0}\in\mathit{T}$, it cannot hold that
  \[\myvector{b}^{T} \exp(At) \myvector{x}_{0}^{r} \equiv 0 \, ,\]
  since $\myvector{b}^{T} \exp(At) \myvector{x}_{0}^{c}$ is negative
  infinitely often by \cref{cor:liminf}.

  Suppose that $\myvector{x}_{0}\in\mathit{T}$ and let $(\rho,m)$
  and $(\eta,j)$ be the dominant indices for
  $\myvector{b}^{T} \exp(At) \myvector{x}_{0}^{r}$ and
  $\myvector{b}^{T} \exp(At) \myvector{x}_{0}^{c}$ respectively.
Then by \cref{prop:linear} we have
\begin{gather}
\myvector{b}^{T} \exp(At) \myvector{x}_{0}^{r} = \exp(\rho t)t^m (c
  + o(1))
\label{eq:real}
\end{gather}
 as $t \rightarrow \infty$, where $c$ is a non-zero real
number.  We will show that $c>0$, from which it follows  that
$\myvector{x}_{0}^{r} \in \mathit{ET}$.



It must hold that $(\eta,j)\preceq (\rho,m)$.  Indeed, if $(\eta,j)\succ
(\rho,m)$, then, as $t\rightarrow\infty$,
\begin{align*}
\myvector{b}^{T} \exp(At) \myvector{x}_{0} = \exp(\eta t)t^{j} \Bigg(
\underbrace{\frac{\myvector{b}^{T} \exp(At) \myvector{x}_{0}^{c}}{\exp(\eta t)t^{j}}}_{\mytag{termA}} + o(1)\Bigg) ,
\end{align*}
but the limit inferior of the term~\ref{termA} above is strictly
negative by \cref{cor:liminf},
contradicting the fact that $\myvector{x}_{0}\in\mathit{T}$.

If $(\eta,j)=(\rho,m)$, then, as $t\rightarrow\infty$,
\begin{align*}
\myvector{b}^{T} \exp(At) \myvector{x}_{0} = \exp(\rho t)t^{m} \left( c + \frac{\myvector{b}^{T} \exp(At) \myvector{x}_{0}^{c}}{\exp(\rho t)t^{m}} + o(1) \right) ,
\end{align*}
and by invoking \cref{cor:liminf} as above, it follows that
$c > 0$.

Finally, if $(\eta,j)\prec (\rho,m)$, then, as $t\rightarrow\infty$,
\begin{gather}\myvector{b}^{T} \exp(At) \myvector{x}_{0}^c =
\exp(\rho t)t^{m} \cdot o(1) \, ,
\label{eq:complex}
\end{gather} and hence, by (\ref{eq:real}) and (\ref{eq:complex}), it
follows that
\begin{align*}
\myvector{b}^{T} \exp(At) \myvector{x}_{0} = \exp(\rho t)t^{m} \left( c +o(1) \right) \, .
\end{align*}
From the fact that $\myvector{x}_{0} \in T$ and that $c\neq 0$ we must have $c>0$.

In all cases it holds that $c>0$ and hence
$\myvector{x}_{0}^{r}\in\mathit{ET}$.



Having argued that
$\mathit{ET}\neq\emptyset$ iff $\mathit{ET}\cap
\mathcal{V}^{r}\neq\emptyset$,
we will now show that the set $\mathit{ET}\cap \mathcal{V}^{r}$ is a
convex polytope that we can efficiently compute. As before, it
suffices to prove this claim for the case when
$\mathcal{P}=\lbrace \myvector{x}\in\Reals^{d}:
\myvector{b}^{T}\myvector{x}\triangleright 0\rbrace$
(where $\triangleright$ is either $>$ or $\geq$).

In what follows, we let $[K]$ denote the set $\lbrace 0,\ldots,K-1 \rbrace$. We can write
\begin{align*}
\myvector{b}^{T}\exp(At)=\sum\limits_{(\eta,j)\in\sigma(A)\times
  [\nu(A)]} \exp(\eta t)t^{j} \myvector{u}_{(\eta,j)}^{T} \, ,
\end{align*}
where $\myvector{u}_{(\eta,j)}^{T}$ is a vector of coefficients.

Note that if $\myvector{x}\in\mathcal{V}^{r}$ and $(\eta,j)\in
(\sigma(A)\setminus \Reals)\times \Naturals_{0}$, then
$\myvector{u}_{(\eta,j)}^{T}\myvector{x}=0$, as
$\myvector{u}_{(\eta,j)}^{T}\myvector{x}$ is the coefficient of $t^{j} \exp(\eta
t)$ in $\myvector{b}^{T}\exp(At) \myvector{x}$, and $\mathcal{V}^{r}$ is
invariant under $\exp(At)$. Moreover,
\begin{align*}
\mathit{ET}\cap\mathcal{V}^{r}=(\mathcal{B}\cap\mathcal{C}) \cup
\begin{cases}
\lbrace \myvector{0} \rbrace & \text{ if } \triangleright \text{ is } \geq \\
\emptyset & \text{ if } \triangleright \text{ is } >
\end{cases}
\end{align*}
where
\begin{align*}
\mathcal{B}=&\bigcap\limits_{(\eta,j)\in (\sigma(A)\setminus \Reals)\times [\nu(A)]} \lbrace \myvector{x}\in\Reals^{d}: \myvector{u}_{(\eta,j)}^{T} \myvector{x}=0 \rbrace \\
\mathcal{C}=&\bigcup\limits_{(\eta,j)\in(\sigma(A)\cap \Reals)\times [\nu(A)]} \bigg[ \lbrace \myvector{x}\in\Reals^{d}: \myvector{u}_{(\eta,j)}^{T} \myvector{x}>0 \rbrace \cap \\
&\bigcap\limits_{(\rho,m)\succ (\eta,j)} \lbrace \myvector{x}\in\Reals^{d}: \myvector{u}_{(\rho,m)}^{T} \myvector{x}=0 \rbrace \bigg]
\end{align*}

The set $\mathit{ET}\cap\mathcal{V}^{r}$ can be seen to be convex from
the above characterisation. Alternatively, note that $\mathit{ET}$ can
be shown to be convex from its definition and that $\mathcal{V}^{r}$
is convex, therefore so must be their intersection. Thus
$\mathit{ET}\cap\mathcal{V}^{r}$ must be a convex polytope whose
definition possibly involves canonically-represented real algebraic
numbers, and the Polytope Escape Problem reduces to testing this
polytope for non-emptiness.
\end{proof}

\section{Conclusion}

We have shown that the Polytope Escape Problem for continuous-time
linear dynamical systems is decidable, and in fact, polynomial-time
reducible to the decision problem for the existential theory of real
closed fields.  Given an instance of the problem $(f,\mathcal{P})$,
with $f$ an affine map, our decision procedure involves analysing the real
eigenstructure of the linear operator
$g(\boldsymbol{x}):=f(\boldsymbol{x})-f(\boldsymbol{0})$. In fact, we
showed that all complex eigenvalues could essentially be ignored for
the purposes of deciding this problem.

Interestingly, the seemingly closely related question of whether a
given single trajectory of a linear dynamical system remains trapped
within a given polytope appears to be considerably more challenging
and is not known to be decidable. In that instance, it seems that the
influence of the complex eigenstructure cannot simply be discarded.

\bibliographystyle{abbrv}
\bibliography{refs}

\end{document}